\documentclass[12pt,onecolumn,draftcls, letterpaper]{IEEEtran}
\makeatletter
\def\ps@headings{%
\def\@oddhead{\mbox{}\scriptsize\rightmark \hfil \thepage}%
\def\@evenhead{\scriptsize\thepage \hfil \leftmark\mbox{}}%
\def\@oddfoot{}%
\def\@evenfoot{}}
\makeatother
\usepackage{url}
\usepackage{nomencl}
\usepackage{graphics,latexsym, graphicx}
\usepackage{amsmath,amssymb,amsthm, amstext}
\usepackage{amsmath}
\usepackage{dsfont}
\usepackage{epsf}
\usepackage[T1]{fontenc}
\usepackage{subfigure}
\usepackage{algorithm}
\usepackage{algorithmic}
\usepackage{url}
\usepackage{color}
\usepackage{algorithm}
\usepackage{algorithmic}

\newtheorem{lemma}{Lemma}

\newtheorem{example}{Example}

\newcommand{\mS}{\mathcal{S}}

\newcommand{\be}{\begin{eqnarray}}
\newcommand{\ee}{\end{eqnarray}}
\newcommand{\ben}{\begin{eqnarray*}}
\newcommand{\een}{\end{eqnarray*}}
\newcommand{\dref}[1]{(\ref{#1})}
\newcommand{\expect}[1]{{\mathbb E} \Bigl[ #1\Bigr]}

\newcommand{\prob}[1]{{\mathbb P} \left( #1\right)}
\newcommand{\mG}{\mathcal{G}}
\newcommand{\mV}{\mathcal{V}}
\newcommand{\mE}{\mathcal{E}}
\newcommand{\mI}{\mathbf{I}}

\begin{document}
\title{Opinion Dynamics in Social Networks: A Local Interaction Game with Stubborn Agents}
\author{\IEEEauthorblockN{Javad Ghaderi and R. Srikant\\}
\IEEEauthorblockA{Department of ECE and Coordinated Science Lab.\\
 University of Illinois at Urbana-Champaign\\
\{jghaderi, rsrikant\}@illinois.edu
}}

\maketitle
\begin{abstract}
The process by which new ideas, innovations, and behaviors spread through a large social network can be thought of as a networked interaction game: Each agent obtains information from certain number of agents in his friendship neighborhood, and adapts his idea or behavior to increase his benefit.  In this paper, we are interested in how opinions, about a certain topic, form in social networks. We model opinions as continuous scalars ranging from $0$ to $1$ with $1$ ($0$) representing extremely positive (negative) opinion. Each agent has an initial opinion and incurs some cost depending on the opinions of his neighbors, his initial opinion, and his stubbornness about his initial opinion. Agents iteratively update their opinions based on their own initial opinions and observing the opinions of their neighbors. The iterative update of an agent can be viewed as a myopic cost-minimization response (i.e., the so-called best response) to the others' actions. We study whether an equilibrium can emerge as a result of such local interactions and how such equilibrium possibly depends on the network structure, initial opinions of the agents, and the location of stubborn agents and the extent of their stubbornness. We also study the convergence speed to such equilibrium and characterize the convergence time as a function of aforementioned factors. We also discuss the implications of such results in a few well-known graphs such as Erdos-Renyi random graphs and small-world graphs.
\end{abstract}
\section{Introduction}\label{sec:intro}
Rapid expansion of online social networks, such as friendships and information networks, in recent years has raised an interesting question: how do opinions form in a social network? The opinion of each person is influenced by many factors such as his friends, news, political views, area of professional activity, and etc. Understanding such interactions and predicting how specific opinions spread throughout social networks has triggered vast research by economists, sociologist, psychologies, physicists, etc.

We consider a social network consisting of $n$ agents. We model the social network as a graph $\mG(\mV,\mE)$ where agents are the vertices and edges indicate pairs of agents that have interactions. Agents have some private initial opinions and iteratively update their opinions based on their own initial opinions and observing the opinions of their neighbors. We study whether an equilibrium can emerge as a result of such local interactions and how such equilibrium possibly depends on the graph structure and initial opinions of the agents. In the interaction model, we also incorporate \emph{stubbornness} of agents with respect to their initial opinions and investigate the dependency of the equilibrium on such stubborn agents. Characterizing the convergence rate to the equilibrium as a function of graph structure, location of stubborn agents and their levels of stubbornness is another goal of the current paper.
\subsection{Related Work}
There has been an interesting line of research trying to explain emergence of new phenomenon, such as spread of innovations and new technologies, based on local interactions among agents, e.g., \cite{kandori, kleinberg2, ellison, saberi}. Roughly speaking, a coordination game is played between the agents in which adopting a common strategy has a higher payoff. Agents behave according to a noisy version of the best-response dynamics. Introducing the noise eliminates the possibility of multiple equilibria and drives the system to a particular equilibrium in which all agents take the same action \cite{kleinberg2, ellison, saberi}. In particular, \cite{ellison} studied the rate of convergence for two scenarios of extreme interaction, namely, a complete graph and a ring graph, and showed that the dynamics converges very slowly in the complete graph and very fast in the ring network. Reference \cite{saberi} studies the convergence time for general networks and establishes similar results that show highly-connected non-local graphs exhibit slow convergence and poorly connected, low dimensional graphs exhibit fast convergence.

There is a rich and still growing literature on social learning using a Bayesian perspective where individuals observe the actions of others and update their beliefs iteratively about an underlying state variable, e. g., \cite{bikchandani, banraji, acemoglu}. There is also opinion dynamics based on non-Bayesian models, e. g., those in \cite{degroot, ellison2, bala, krause, borkar, parandeh}. In particular, \cite{parandeh} investigates a model in which agents meet and adopt the average of their pre-meeting opinions and there are also forceful agents that influence the opinions of others but may not change their opinions. Under such a model, and assuming that even forceful agents update their opinions when meeting some agents, \cite{parandeh} investigates convergence to the average of the initial opinions and characterizes the amount of divergence from the average due to such forceful agents. As reported in \cite{parandeh}, it is significantly more difficult to analyze social networks with several forceful agents that do not change their opinions and requires a different mathematical approach. Our model is closely related to the non-Bayesian framework, this keeps the computations tractable and can characterize the equilibrium in presence of agents that are biased towards their initial opinions (the so-called partially stubborn agents in our paper) or do not change their opinions at all (the so-called fully stubborn agents in our paper). The recent work~\cite{yildiz} studies opinion dynamics based on the so-called voter model where each agent holds a binary $0$-$1$ opinion and at each time a randomly chosen agent adopts the opinion of one of his neighbors, and there are also stubborn agents that do not change their states. Under such model, \cite{yildiz} shows that the opinions converge in distribution and characterizes the first and the second moments of this distribution. In addition, our paper is also related to consensus problems in which the question of interest is whether beliefs (some scalar numbers) held by different agents will converge to a common value, e.g., \cite{tsitsikilis2, tsitsikilis3, tsitsikilis, jadbabaie, fagnani}.
\subsection{Contributions}
Our first set of results, Lemmas \ref{lemma1} and \ref{conv1}, are rather straightforward characterizations of equilibrium and convergence time when there are no stubborn agents. These results are analogs of well-known convergence results for the probability distribution over the states of a Markov chain. In this case, our model reduces to a \emph{continuous coordination game} and the (noisy) best-response dynamics converge to a common opinion in which the impact of each agent is directly proportional to his degree in the social network. The analysis of the convergence speed, in such a continuous coordination game, reveals conclusions that are different from those in \cite{ellison}, \cite{saberi} in the context of a two-strategy coordination game. The convergence of best-response dynamics in highly connected non-local graphs, like the complete graph, occurs at a faster speed than the convergence in poorly connected local graphs like the ring graph.

Our second set of results are concerned with social networks in which some of the agents are fully/partially stubborn. In this case, the best-response dynamics at each agent converges to a convex combination of the initial opinions of the stubborn agents (Lemma \ref{lemma: convex}). The impact of each stubborn agent on such an equilibrium is related to appropriately defined hitting probabilities over a modified graph $\hat{\mG}(\hat{\mV}, \hat{\mE})$ of the original social network ${\mG}(\mV, \mE)$. We also give an interesting electrical network interpretation of the equilibrium (Lemma \ref{elec}). Since the exact characterization of convergence time is difficult, we derive appropriate upper-bounds (Lemma \ref{diaconis}, Lemma \ref{sinclaire}) and a lower-bound (Lemma \ref{lemma: lower}) on the convergence time that depend on the structure of the social network (such as the diameter of the graph and the relative degrees of stubborn and non-stubborn agents), and the location of stubborn agents and their levels of stubbornness. Based on such bounds, we study the convergence speed in social networks with different topologies such as expander graphs, Erdos-Renyi random graphs, and small-world networks.

Finally, in concluding remarks, we discuss the implication of our results in applications where limited advertising budget is to be used to convince a limited number of agents in social networks to adopt, for example, a certain opinion about a product/topic. Such agents may in turn convince others to change their opinions. Our results shed some light on the optimal selection of such agents to trigger a faster spread of the advertised opinion throughout the social network.
\subsection{Organization}
The organization of the paper is as follows. We start with the definitions and introduce our model in Section \ref{model}. In Section \ref{no stubborn}, we investigate opinion dynamics in social networks without any stubborn agents. We consider social networks with at least one stubborn agent in Section \ref{with stubborn}. Section \ref{conclusions} contains our concluding remarks. The proofs of the results are provided in the appendices at the end of the paper.
\subsection{Basic notations}
All the vectors are column vectors. $x^T$ denotes the transpose of vector $x$. A diagonal matrix with elements of vector $x$ as diagonal entries is denoted by $\mathrm{diag}(x)$. $x_{max}$ means the maximum element of vector $x$. Similarly, $x_{min}$ is the minimum element of vector $x$. $\mathds{1}_n$ denotes a vector of all ones of size $n$.
\section{Model and Definitions}\label{model}
Consider a social network with $n$ agents, denoted by a graph $\mG(\mV,\mE)$ where agents are the vertices and edges indicate the pairs of agents that have interactions. For each agent $i$, define its neighborhood $\partial_i$ as the set of agents that node $i$ interacts with, i.e., $\partial_i:=\{j: (i,j)\in \mE\}$. Each agent $i$ has an initial opinion $x_i(0)$. For simplicity we assume that the initial opinions are some numbers between $0$ and $1$. For example, $x_i(0)$'s could represent the people opinions about the economic situation of the country, ranging from $0$ to $1$ with an opinion $1$ corresponding to perfect satisfaction with the current economy and $0$ representing an extremely negative view towards the economy. Let $x(0):=[x_1(0) \cdots x_n(0)]^T$ denote the vector of initial opinions. We assume each agent $i$ has a cost function of the form
\be
J_i(x_i, x_{\partial_i})=\frac{1}{2}\sum_{j \in \partial_i}(x_i-x_j)^2+\frac{1}{2}K_i(x_i-x_i(0))^2,
\ee
that he tries to minimize where $K_i\geq 0$ measures the \textit{stubbornness} of agent $i$ regarding his initial opinion\footnote{Although we have considered uniform weights for the neighbors, the results in the paper hold under a more general setting when each agent puts a weight $w_{ij}$ for his neighbor $j$.}.
When none of the agents are stubborn, correspondingly $K_i$'s are all zero, the above formulation defines a \textit{coordination game} with continuous payoffs because any vector of opinions $x=[x_1 \cdots x_n]^T$ with $x_1=x_2=\cdots=x_n$ is a \textit{Nash equilibrium}. Here, we consider a synchronous version of the game between the agents. At each time, every agent observes the opinions of his neighbors and updates his opinion based on these observations and also his own initial opinion in order to minimize his cost function. It is easy to check that, for every agent $i$, the best-response strategy is
\be \label{individual dynamics}
x_i(t+1)=\frac{1}{d_i+K_i}\sum_{j \in \partial_i}x_j(t)+\frac{K_i}{d_i+K_i}x_i(0),
\ee
where $d_i=|\partial_i|$ is the degree of node $i$ in graph $\mG$. Define a matrix $A_{n \times n}$ such that $A_{ij}=\frac{1}{d_i+K_i}$ for $(i,j) \in \mE$ and zero otherwise. Also define a diagonal matrix $B_{n \times n}$ with $B_{ii}=\frac{K_i}{d_i+K_i}$ for $1 \leq i \leq n$. Thus, in the matrix form, the best response dynamics are given by
\be \label{matrix dynamics}
x(t+1)=Ax(t)+Bx(0).
\ee
Iterating (\ref{matrix dynamics}) shows that the vector of opinions at each time $t\geq 0$ is
\be \label{solution}
x(t)=A^tx(0)+\sum_{s=0}^{t-1}A^{s}Bx(0).
\ee
In the rest of the paper, we investigate the existence of equilibrium under the dynamics (\ref{matrix dynamics}) in different social networks, with or without stubborn agents. We also characterize the convergence time of the dynamics, i.e., the amount of time that it takes for the agents' opinions to get close to the equilibrium. The equilibrium behavior is relevant only if the convergence time is reasonable \cite{ellison}.
\section{No Stubborn Agents}\label{no stubborn}
Convergence issues in the case of no stubborn agents is a special case of consensus, and has been well studied. Here, we briefly review this work to put our later results in context. This also allows us to compare the results in \cite{ellison} to continuous opinion dynamics.

When there are no stubborn agents in the social network, i.e., all $K_i$'s are zero, $A$ is a row-stochastic matrix and $B=0$ in (\ref{matrix dynamics}). Without loss of generality, assume that $A$ is irreducible which corresponds to $\mG$ being a connected graph. (Otherwise, $A$ can be simply viewed as a collection of diagonal sub-matrices with each sub-matrix describing the opinion dynamics on one of the disconnected subgraphs of $\mG$.). For now, further assume that $A$ is primitive, i.e., there is a constant $t_0$ such that all elements of $A^t$ are strictly positive for all $t \geq t_0$. It is easy to show that $A$ is primitive if and only if the graph $\mG$ is not bipartite, i.e., there is both an odd cycle and an even cycle from every node to itself (we will later discuss the case of bipartite graphs).
\subsection{Existence and characterization of the equilibrium}
By the Perron-Frobenius theorem, the eigenvalues of $A$ are such that $1=\lambda_1 > |\lambda_2| \geq \cdots \geq |\lambda_n|$. The right eigenvector corresponding to $\lambda_1=1$ is $ \mathds{1}_n$ up to a normalization constant.
Let $\pi$ be the left eigenvalue of $A$ which is unique up to a normalization constant.

Using the eigen-decomposition of A (or Jordan normal decomposition if eigenvalues are not distinct), it is well-known that $\lim_{t \to \infty} A^t =\mathds{1}_n\pi^T$ up to a scalar multiple. In fact, since $A^t$ is row-stochastic, $\sum_{i=1}^n \pi_i=1$. Hence, $\pi$ can be interpreted as the unique stationary distribution of a Markov chain with transition probability matrix $A$, so $A^t$ must converge to the matrix with all rows equal to $\pi^T$. Hence, the equilibrium under dynamics (\ref{matrix dynamics}) is unique and simply given by
\be
x({\infty}):=\lim _{t \to \infty}x(t) = \mathds{1}_n\pi^Tx(0).
\ee
This shows that in equilibrium, agents will reach a \textit{consensus}, i.e., all opinions are eventually the same and equal where $x_i(\infty)=\sum_{i=1}^n\pi_ix_i(0)$, for all $i \in \mV$.

Note that $A$ can be interpreted as transition probability matrix of a random walk over the graph $\mG$ with edge weights equal to one. It is easy to check that such random walk is reversible and the stationary distribution of the random walk (i.e., left eigenvector of $A$ corresponding to the eigenvalue $1$) is simply $\pi_i=d_i/2 |\mE|$, $i \in \mV$. Hence, the impact of each agent on the equilibrium consensus is directly proportional to its degree. We state the result as the following lemma.
\begin{lemma}\label{lemma1}
In a social network $\mG$ with no stubborn agents, and initial opinions vector $x(0)$, the best response dynamics will converge to the following unique equilibrium
\be
x_i(\infty)=\frac{1}{2|\mE|}\sum_{j=1}^nd_j x_j(0);\mbox{ for all }i \in \mV,
\ee
where $d_i$ is the degree of agent $i$.
\end{lemma}
\subsection{Convergence time}
First, we introduce a convenient norm on $\mathbb{R}^n$ that is linked to the stationary distribution $\pi$ of the random walk on the social network graph. Let $\ell^2(\pi)$ be the real vector space $\mathbb{R}^n$ endowed with the scalar product
$$
\langle z,y\rangle_{\pi}:=\sum_{i=1}^r z(i)y(i)\pi(i).
$$
Then, the norm of $z$ with respect to $\pi$ is defined as
$$
\|z\|_{\pi}:=\left(\sum_{i=1}^r z(i)^2\pi(i)\right)^{1/2}.
$$

Define the error as the vector
\be
e(t):=x(t)-x(\infty).
\ee
The following lemma states that the error goes to zero geometrically at a rate equal to the second largest eigenvalue modulus of $A$. 
\begin{lemma}\label{conv1}
Under the best-response dynamics,
\be
\|e(t)\|_{\pi}\leq \rho_2^t \|e(0)\|_{\pi}.
\ee
where $\rho_2:=\max_{i\neq 1} {|\lambda_i|}$ is the SLEM (Second Largest Eigenvalue Modulus) of $A$.
\footnote{
In Euclidian norm, $\sqrt{\pi_{min}}\|e(t)\|_2 \leq  \|e(t)\|_{\pi}\leq \sqrt{\pi_{max}}\|e(t)\|_2$, and so
$
\|e(t)\|_2\leq \rho_2^t \sqrt{\frac{d_{max}}{d_{min}}}\|e(0)\|_2,
$
where $d_{max}=\max _id_i$ and $d_{min}=\min_i d_i$.
}
\end{lemma}
See Appendix \ref{proof-lemma-conv1} for the proof. We define the convergence time $\tau(\nu)$ as
\be
\tau(\nu)=\inf\{t \geq 0: \|e(t)\|_{\pi} \leq \nu\},
\ee
where $0 < \nu \ll 1 $ is some small positive number.

It can be seen that under the $\pi$-norm, $\|e(t)\|_{\pi}=\mathrm{Var}_{\pi}(e(t))$, i.e., the variance of $e(t)$ with respect to distribution $\pi$. This ensures that $\|e(0)\|_{\pi} \leq 1$ even when $n \to \infty$ because we assumed that initial opinions are bounded between zero and one. A simple calculation, based on Lemma \ref{conv1}, reveals that
\ben
\left(\frac{1}{1-\rho_2}-1\right)\log\left(\frac{\|e(0)\|_{\pi}}{\nu}\right) \leq \tau(\nu) \leq \frac{1}{1-\rho_2}\log\left( \frac{\|e(0)\|_{\pi}}{\nu}\right).
\een
In particular, the convergence time is $\Theta\Big(\frac{1}{1- \rho_2}\Big)$ as the number of agents $n$ grows.
\subsection{Bipartite networks and noisy opinion dynamics}
Next, we consider the case that the social network is bipartite. One important example of such networks is a ring network \textit{with an even number of agents $n$}. The ring graph is formed by placing the agents on a circle and connecting each agent to two of his nearest neighbors. In this case, the best-response dynamics do not converge to an equilibrium. For example, for the case of the ring network with $n$ even, matrix $A$ is simply
\ben
A_{ij}= \left\{\begin{array}{ll}
1/2& \mbox{if } |i-j|= 1\mbox{ or } n-1\\
0& \mbox{otherwise}.
\end{array}\right.
\een
Then, it is easy to see that, as $t \to \infty$, $A^t$ alternates between two matrices for $t$ even and $t$ odd, in fact, for $t$ odd,
\ben
\lim_{t \to \infty} A^{(t)}_{ij}= \left\{\begin{array}{ll}
2/n& \mbox{if } |i-j|  \equiv1 \mod 2,\\
0& \mbox{otherwise}
\end{array}\right.
\een
and for $t$ even,
\ben
\lim_{t \to \infty}A^{(t)}_{ij}= \left\{\begin{array}{ll}
2/n& \mbox{if } |i-j| \equiv0 \mod 2\\
0& \mbox{otherwise}.
\end{array}\right.
\een
This shows that, as $t \to \infty$, the opinion of each agent $i$ does not converge and alternates between two values. The opinion of agent $i$ at an odd $t$ will be the average of the initial opinions of the agents $\{j: |i-j| \equiv 1 \mod 2\}$ and at even $t$, it will be the average of the initial opinions of the agents $\{j: |i-j| \equiv 0 \mod 2\}$.

In practice, not everyone completely ignores his own previous opinion and might be slightly biased by his old opinion. Hence, we can consider a noisy version of the best-response dynamics as follows
 \be
x_i^{(\epsilon)}(t+1)=(1-\epsilon)\Big(\frac{1}{d_i}\sum_{j \in \partial_i}x_j(t)\Big)+\epsilon x_i(t),
\ee
for some \textit{self-confidence} $\epsilon>0$. Here, we assume all agents have the same self-confidence but the argument can be adapted for different self-confidences as well. Introducing such self-confidences, adds self-loops to graph $\mG$ which ensures that $\mG$ is not bipartite, or, correspondingly, $A^{(\epsilon)}$ is primitive where $A_{ij}^{(\epsilon)}=\epsilon$ if $j=i$ and $A_{ij}^{(\epsilon)}=(1-\epsilon)/d_i$ if $j \in \partial_i$. Hence, using the results in the previous section, the noisy best-response dynamics will converge to
\ben
x_i^{(\epsilon)}(\infty)=\sum_{j=1}^n\pi_j^{(\epsilon)}x_j(0); \mbox{ for all }i \in \mV,
\een
where $\pi^{(\epsilon)}=[\pi_1^{(\epsilon)} \cdots \pi_n^{(\epsilon)} ]^T$ is the unique stationary distribution of the Markov chain\footnote{The terminologies random walk and Markov chain can be used interchangeably here.} with transition probability matrix $A^{(\epsilon)}$. By reversibility of $A^{(\epsilon)}$, $\pi_i^{(\epsilon)}=\pi_i=\frac{d_i}{2|\mE|}$ independently of $\epsilon$. Hence,
$$x_i^{(\epsilon)}(\infty)= x_i(\infty)=\frac{1}{2|\mE|}\sum_{j=1}^nd_jx_j(0), \mbox{ for all i }\in \mV,$$ i.e., agents will converge to the same equilibrium as in the non-partite case.

To investigate the convergence rate, note that $A^{(\epsilon)}=\epsilon \mI+(1-\epsilon)A$, so $\lambda_i^{(\epsilon)}=\epsilon+(1-\epsilon)\lambda_i(A)$. Especially, $\lambda_2^{(\epsilon)}=\epsilon+(1-\epsilon)\lambda_2(A)$ and $\lambda_n^{(\epsilon)}=-1+2\epsilon$ because $\lambda_n(A)=-1$ for random walk on bipartite graph. Hence, as far as the scaling law with $n$ concerns, the convergence time is $\Theta\Big(\frac{1}{1-\rho_2^{(\epsilon)}}\Big)$ for $\rho_2^{(\epsilon)}=\max\{\lambda_2^{(\epsilon)}, |\lambda_n^{(\epsilon)}|\}$, which, for fixed $\epsilon$, means that the convergence time is  $\Theta\Big(\frac{1}{1- \lambda_2(A)}\Big)$.
Similarly, we can consider the noisy best-response dynamics for the non-partite graphs. Again, introducing self-confidence in opinion dynamics does not change the equilibrium. Moreover, the second largest and the smallest eigenvalues will be respectively given by $\lambda_2^{(\epsilon)}=\epsilon+(1-\epsilon)\lambda_2(A)$ and $\lambda_n^{(\epsilon)}> -1+2\epsilon$. Hence, the convergence time is again of the order $\Theta\Big(\frac{1}{1- \lambda_2(A)}\Big)$ as $n$ grows.

Hence, the convergence time of the best-response dynamics is determined by $\lambda_2(A)$ of the corresponding random walk over the social graph.
\begin{example}
In this example, we make a comparison between ring graph and complete graph with $n$ nodes. The complete graph represents situation when all agents can communicate with each other with no constrains, while in the ring graph, each agent can only communicate with his two nearest neighbors. Qualitatively, the complete graph and the ring graph represent two extreme ends of the spectrum of graphs~\cite{shah}. Note in both cases the (noisy) best-response dynamics converge to the average of initial opinions. It is easy to see that $\lambda_2(A)$ is $\frac{1}{n-1}$ for the complete graph and $\cos(\frac{2\pi}{n})\approx 1-\frac{\pi}{n^2}$ for the ring. Hence, while both of the graphs have the same equilibrium, convergence in the complete graph is much faster than convergence in the ring, in fact, $O(1)$ vs. $O(n^2)$.
\end{example}
It is interesting to compare our continuous coordination game with a two-strategies coordination game in Ellison \cite{ellison} where agents behave according to a noisy best-response dynamics. Introducing the noise eliminates the possibility of multiple equilibria and leads to convergence to the ``risk dominant`` strategy throughout the network. He studied the rate of convergence for the complete graph and the ring graph and showed that the dynamics converges very slowly in the complete graph and very fast in the ring network. This is in total contrast with the above example as we observe that, in a continuous coordination game, the best-response dynamics converge faster in the complete graph compared to the ring graph. This can be justified by noticing that the mechanism for spreading a common strategy throughout the network is inherently different in the continuous and the two-strategy coordination game. In the two-strategy coordination game, existence of a sufficiently large cluster of agents playing the risk-dominant strategy needed in order to prevail the risk-dominant strategy throughout the network. As Elision stated, in the ring network, the required size of such clusters is very smaller than the required size of clusters for the complete graph. Hence, it takes an extremely long time to see such clusters in the complete graph starting from arbitrary initial conditions. In the continuous coordination game, the (noisy) best-response dynamics converge to the average of initial opinions in both graphs. In the complete graph, after the first iteration, the opinion of each node is the average of the initial opinions of all agents excluding its own opinion which has little importance especially when $n$ is large. In the ring graph, averaging the initial opinions of an agents's neighbors could still be very far from the average of all the initial opinions and increasing $n$ will make this gap even larger.

There is a rich literature on approximating $\lambda_2(A)$ of a random walk over different types of graphs, e.g., see Chapter 2 of \cite{shah} for a survey. Intuitively, the convergence time is dominated by the highly connected component of the graph which is loosely connected to the rest of the network (captured by the notion of \textit{conductance} of the Markov chain \cite{sin}, or the edge isoperimetric function of the graph). For example, expander graphs have fast convergence (with a convergence time independent of $n$) because the number of connections from every subgraph, of size less than $n/2$, to the rest of the network is at least a constant fraction of total connections within the subgraph independent of $n$. We do not proceed in this direction further and in the next section, we study the more interesting case of social networks with stubborn agents.
\section{Impact of Stubborn Agents}\label{with stubborn}
\subsection{Existence and characterization of equilibrium}\label{sec: stubborn-existence}
Consider a connected social network $\mG(\mV,\mE)$ in which at least one of the agents is stubborn, i. e., $K_i > 0$ for some $i \in \mV$. Then $A$ is an irreducible sub-stochastic matrix with the row-sum of at least one row less than one. Let $\rho_1(A):=\max_i |\lambda_i(A)|$ denote the spectral radius of $A$. It is well-known that $\rho_1(A)$ of a sub-stochastic matrix $A$ is less than one, and hence, $\lim_{t \to \infty} A^t=0$. Therefore, by Perron-Ferobenius theorem, the largest eigenvalue should be positive, real $1> \lambda_1 > 0$ and $\rho_1(A)=\lambda_1$. Hence, in this case, based on (\ref{solution}), the equilibrium exists and is equal to
\be \label{equi2}
x({\infty}):=\lim _{t \to \infty}x(t) =\sum_{s=0}^\infty A^s Bx(0)=(I-A)^{-1}Bx(0).
\ee
Therefore, since $B_{ii}=0$ for all non-stubborn agents $i$, the initial opinions of non-stubborn agents will vanish eventually and have no effect on the equilibrium (\ref{equi2}).

The matrix form \dref{equi2} does not give any insight on how the equilibrium depends on the graph structure and the stubborn agents. Next, we describe the equilibrium in terms of explicit quantities that depend on the graph structure, location of stubborn agents and their levels of stubbornness.

Let $\mS \subseteq \mV$ be the set of stubborn agents and $|\mS|\geq 1$. Any agent $i$ in $\mS$ is either \textit{fully stubborn}, meaning its corresponding $K_i = \infty$, or it is \textit{partially stubborn}, meaning $0< K_i < \infty$. Hence, $\mS=\mS_F \cup \mS_P$ where $\mS_F$ is the set of fully stubborn agents and $\mS_P$ is the set of partially stubborn agents\footnote{We need to distinguish between the case $0 <K_i < \infty$ and $K_i=\infty$ for technical reasons; however, as it will become clear later, the conclusions for $K_i=\infty$ are equivalent to those for $K_i < \infty$ if we let $K_i \to \infty$}.
Without loss of generality, index the partially stubborn agents with $1, \cdots, |\mS_P|$, index the fully stubborn agents with $|\mS_P|+1, \cdots, |\mS|$ and finally the non-stubborn agents with $|\mS|+1, \cdots, n$.
Next, we construct a \textit{weighted graph} $\hat{\mG}(\hat{\mV}, \hat{\mE})$ based on the original social graph $\mG(\mV, \mE)$ and the location of partially stubborn agents $\mS_P$ and their levels of stubbornness $K_i$, $i \in \mS_P$.

Assign weight $1$ to all the edges of $\mG$. Connect a new vertex $u_i$ to each $i \in \mS_P $ and assign a weight $K_i$ to the corresponding edge. Index the new vertex connected to $i \in \mS_P$ by $n+i$. We use the node $u_i$ and its index $n+i$ interchangeably. Let $\hat{\mV}:=\mV \cup \{u_i:i \in \mS_P \}=\{1, 2, \cdots, n+|\mS_P|\}$ and $\hat{\mE}:=\mE \cup \{(i,u_i): i \in \mS_P \}$. Also let $w_{ij}$ denote the weight of edge $(i,j)\in \hat{\mE}$. Then $\hat{\mG}(\hat{\mV}, \hat{\mE})$ is a weighted graph with weights $w_{ij}=1$ for all $(i,j) \in \mE$ (the edges of $\mG$) and $w_{iu_i}=K_i$ for all $ i \in \mS_P$. Let $u(\mS_P)= \{u_i: i \in \mS_P\}$.

Define $w_i:=\sum_{j: (i,j) \in \hat{\mE}}w_{ij}$ as the \textit{weighted degree} of vertex $i \in \hat{\mV}$. It should be clear that
\be \label{weight}
w_i= \left\{\begin{array}{ll}
{d_i+K_i}&\mbox{ for }  i \in \mS_P,\\
{d_i}&\mbox{ for } i \in \mV \backslash \mS_P,\\
{K_j}& \mbox{ for }i=u_j,  j \in \mS_P.
\end{array} \right.
\ee
Consider the random walk $Y(t)$ over $\hat{\mG}$ where the probability of transition from vertex $i$ to vertex $j$ is $P_{ij} =\frac{w_{ij}}{w_i}$. Assume the walk starts from some initial vertex $Y(0)=i \in \mV$. For any $j\in \hat{\mV}$ define
\be
\tau_j:=\inf \{t \ge 0: Y(t)=j\},
\ee
as the first hitting time to vertex $j$. Also define $\tau:=\bigwedge_{j \in \mS_F \cup u(\mS_P)} \tau_{j}$ as the first time that the random walk hits any of the vertices in $\mS_F \cup u(\mS_P)$. The following Lemma characterizes the equilibrium. The proof is provided in Appendix \ref{proof-lemma-convex}. 
\begin{lemma}\label{lemma: convex}
The best-response dynamics converge to a unique equilibrium where the opinion of each agent is a convex combination of the initial opinions of the stubborn agents. Based on the random walk over the graph $\hat{\mG}$,
\be \label{equi3}
x_i(\infty)=\sum_{j \in \mS_P} \mathbb{P}_i(\tau=\tau_{u_j}) x_j(0)+ \sum_{j \in \mS_F} \mathbb{P}_i(\tau=\tau_{j}) x_j(0) ; \mbox{ for all } 1 \leq i \leq n,
\ee
where $\mathbb{P}_i(\tau=\tau_{k})$, $ k \in \mS_F \cup u(\mS_P)$, is the probability that the random walk hits vertex $k$ first, among vertices in $\mS_F \cup u(\mS_P)$, given the random walk starts from vertex $i$.
\end{lemma}
Note that $\lim _{K_i \to \infty} \mathbb{P}_i(\tau=\tau_{u_i})=1$ for any partially stubborn agent $ i \in \mS_P$. This intuitively makes sense because as an agent $i$ becomes more stubborn, his opinion will get closer to his own opinion and behaves similarly to a fully stubborn agent.

It should be clear that when there is only one stubborn agent or there are multiple stubborn agents with identical initial opinions, eventually the opinion of every agent will converge to the same opinion as the initial opinion of the stubborn agents.

In general, to characterize the equilibrium, one needs to find probabilities $\mathbb{P}_i(\tau=\tau_{k})$, $ k \in \mS_F \cup u(\mS_P)$. Such hitting probabilities have an interesting electrical network interpretation (see Chapter 3 of \cite{aldous}) as follows. Let $\hat{\mG}$ be an electrical network where each edge $(i,j)\in \hat{\mE}$ has a conductance $w_{ij}$ (or resistance $1/w_{ij}$). Then $\mathbb{P}_i(\tau=\tau_{k})$ is the voltage of node $i$ in the electrical network where node $ k \in \mS_F \cup u(\mS_P)$ is a fixed voltage source of $1$ volt and nodes $ \mS_F \cup u(\mS_P) \backslash \{k\}$ are grounded (zero voltage). This determines the contribution of the voltage source $k$ where all the other sources are turned off. Now let vertices $\mS_F \cup u(\mS_P)$ be fixed voltage sources where the voltage of each source $i \in \mS_F$ is $x_i(0)$ volts and the voltage of each source $u_j \in  u(\mS_P)$, $j \in \mS_P$, is $x_j(0)$ volts. By the linearity of the electrical networks (the superposition theorem in circuit analysis), the voltage of each node in such an electrical network equals to the sum of the responses caused by each voltage source acting alone, while all other voltage sources are grounded. Therefore, the opinion of agent $i$, at equilibrium (\ref{equi3}), is just the voltage of node $i$ in the electrical network model. We mention the result as the following lemma and will prove it directly in Appendix \ref{proof-lemma-elec}.
\begin{lemma}\label{elec}
Consider $\mG$ as an electrical network where the conductance of each edge is $1$ and each stubborn agent $i$ is a voltage source of $x_i(0)$ volts with an internal conductance $K_i$. Fully stubborn agents are ideal voltage sources with infinite internal conductance (zero internal resistance). Then, under the best-response dynamics, the opinion of each agent at equilibrium is just its voltage in the electrical network.
\end{lemma}
We illustrate the use of the above lemma through the following example.
\begin{example}
Consider a one-dimensional social graph, where agents are located on integers $1 \leq i \leq n$. Assume nodes $1$ and $n$ are stubborn with initial opinions $x_1(0)$ and $x_n(0)$, and stubbornness parameters $K_1>0$ and $K_n>0$. Then graph $\hat{\mG}$ consists of the original line network, where each edge weight is one, and two extra nodes $u_1$ and $u_n$ connected to $1$ and $n$ with edge weights $K_1$ and $K_n$. Using the electrical network model, the current is the same over all edges and equal to $I=(x_1(0)-x_n(0))(\frac{1}{K_1}+\frac{1}{K_n}+n-1)^{-1}$. Hence,
$x_i(\infty)=v_i=x_1(0)-I(\frac{1}{K_1}+i-1),$
for $1 \leq i \leq n$, i. e.,
$$
x_i(\infty)=\left(1-\frac{K_1^{-1}+i-1}{K_1^{-1}+K_n^{-1}+n-1}\right)x_1(0)+\left(\frac{K_1^{-1}+i-1}{K_1^{-1}+K_n^{-1}+n-1}\right)x_n(0)
$$
As $K_1$ increases, the final opinion of $i$ will get closer to stubborn agent $1$, and as $K_n$ increases, it will get closer to the opinion of agent $n$.
\end{example}
\subsection{Convergence Time}
Although we are able to characterize the equilibrium in all cases, the equilibrium makes sense only if the time needed to converge to equilibrium is reasonable. Next, we characterize convergence time in the case that there is at least one stubborn agent. Let $e(t)=x(t)-x(\infty)$ be the error vector as defined before. Trivially $e_i(t)=0$ for all fully stubborn agents $i \in \mathcal{S}_F$. Let $\tilde{e}(t):=[e_i(t): i \in \mV \backslash \mS_F]^T$ denote the errors for all other agents. The convergence to the equilibrium \dref{equi2} is geometric with a rate equal to largest eigenvalue of $A$ as stated by the following lemma whose proof is provided in Appendix \ref{proof-lemma-conv2}.
\begin{lemma}\label{conv2}
Let $\tilde{\pi}=[\frac{w_i}{Z}: i \in \mV \backslash \mS_F]^T$ for the weights $w_i$ as in \dref{weight} and $Z$ be the normalizing constant such that $\sum_{i\in \mV \setminus \mS_F} \tilde{\pi}_i=1$. Then,
\be \label{conv}
\|\tilde{e}(t)\|_{\tilde{\pi}}\leq (\lambda_A)^t \|\tilde{e}(0)\|_{\tilde{\pi}},
\ee
where $\lambda_A$ is the largest eigenvalue of $A$.
\end{lemma}
Note that in Euclidian norm, $\sqrt{\tilde{\pi}_{min}}\|\tilde{e}(t)\|_2 \leq  \|\tilde{e}(t)\|_{\tilde{\pi}}\leq \sqrt{\tilde{\pi}_{max}}\|\tilde{e}(t)\|_2$, thus,
$$
\|e(t)\|_2\leq (\lambda_A)^t \sqrt{\frac{w_{max}}{w_{min}}}\|e(0)\|_2,
$$
where $w_{max}:=\max_{i \in \mV \backslash \mS_F}w_i$ and $w_{min}:=\min_{i \in \mV \backslash \mS_F}w_i$.

Hence, the convergence is geometric with a rate at least equal to largest eigenvalue of $A$. Defining the convergence time as the familiar form
$
\tau(\nu): =\inf\{t \geq 0: \|\tilde{e}(t)\|_{\tilde{\pi}} \leq \nu\}
$
for some fixed $\nu >0$, we have
\ben
\left(\frac{1}{1-\lambda_A}-1\right)\log\left(\frac{\|\tilde{e}(0)\|_{\tilde{\pi}}}{\nu}\right) \leq \tau(\nu) \leq \frac{1}{1-\lambda_A}\log\left(\frac{\|\tilde{e}(0)\|_{\tilde{\pi}}}{\nu}\right),
\een
so again $\tau(\nu)=\Theta\left(\frac{1}{1-\lambda_A}\right)$ as $n$ grows. Let $T:= \frac{1}{1-\lambda_A}$. With a little abuse of terminology, we also call $T$ the convergence time.

Since the exact characterization of $\lambda_A$ is difficult, we will derive appropriate upper-bounds and lower-bounds for it that depend on the graph structure, the location of stubborn agents and their levels of stubbornness. The techniques used in deriving the bounds here are similar to the techniques used in deriving geometric bounds for the second largest eigenvalue of stochastic matrices \cite{diaconis, sin2}.

Consider the weighted graph $\hat{\mG}(\hat{\mV}, \hat{\mE})$ as defined in Section \ref{sec: stubborn-existence}. A path from a vertex $i$ to another vertex $j$ in $\hat{\mG}$ is a collection of \textit{oriented} edges that connect $i$ to $j$. For any vertex $i \in \mV \backslash \mS_F$, consider a path $\gamma_{i}$ from $i$ to the set $\mS_F\cup u(\mS_P)$ that does not intersect itself, i.e., $\gamma_i=\{(i,i_1), (i_1,i_2), \cdots, (i_m,j)\}$ for some $j \in \mS_F\cup u(\mS_P)$.

Proceeding along the lines of Diaconis-Stroock \cite{diaconis}, we get the following bound that yields an upper-bound on the convergence time (see Appendix \ref{proof-lemmas} for the proof).
\begin{lemma}\label{diaconis}
 Consider the weighted graph $\hat{\mG}$. Given a set of paths $\{\gamma_i: i \in \mV \backslash \mS_F\}$, from $\mV \backslash \mS_F$ to $\mS_F \cup u (\mS_P)$, let $|\gamma_i|_w:=\sum_{(s,t)\in \gamma_i} \frac{1}{w_{st}}.$ Then, the convergence time $T \leq 2 \xi$ where
 $$
 \xi:=\max_{(x,y)\in \hat{\mE}} \xi(x,y),
 $$
 and, for each edge $(x,y)\in \hat{\mE}$,
\be
\xi(x,y):=\sum_{i: \gamma_i \ni (x,y) } w_i|\gamma_i|_w.
\ee
\end{lemma}
It is also possible to proceed along the lines of Sinclair \cite{sin2}. This gives a different bound stated in the following lemma.
\begin{lemma}\label{sinclaire}
Consider the weighted graph $\hat{\mG}(\hat{\mV}, \hat{\mE})$. Given a set of paths $\{\gamma_i: i \in \mV \backslash \mS_F\}$ from $\mV \backslash \mS_F$ to $\mS_F \cup u (\mS_P)$, we have $T \leq 2 \eta$ where
\be
\eta:=\max_{(x,y)\in \hat{\mE}}\eta(x,y),
\ee
and, for each edge $(x,y)\in \hat{\mE}$,
 \be
 \eta(x,y):=\frac{1}{w_{xy}}\sum_{i:\gamma_i \ni(x,y)} w_i|\gamma_i|.
\ee
\end{lemma}
The above lemma is very similar to the bound reported in \cite{saberi} without proof but differs by a factor of $2$. The factor $2$ is not important in investigating the order of the convergence time; however, in graphs with finite number of agents, ignoring this factor yields convergence times that are smaller than the actual convergence time. Therefore, we have included a short proof in Appendix \ref{proof-lemmas} for the above lemma.

Intuitively, both $\xi(x,y)$ and $\eta(x,y)$ are measures of \textit{congestion} over the edge $(x,y)$ due to paths that pass through $(x,y)$. In general, computing the upper-bound using Lemma \ref{sinclaire} is easier than using Lemma \ref{diaconis}.

An upper bound on $1-\lambda_A$, and thus a lower-bound on the convergence time $T$, is given by the following lemma whose proof is provided in Appendix \ref{proof-lemmas}
\begin{lemma} \label{lemma: lower}
Consider the weighted graph $\hat{\mG}(\hat{\mV}, \hat{\mE})$, then
\be \label{conductance}
{1-\lambda_A} \leq \min_{B\subseteq \mV \backslash \mS_F} \psi(B; \hat{\mG}),
\ee
where $\psi(B; \hat{\mG}):=\frac{\sum_{i \in B, j \notin B}w_{ij}}{\sum_{i \in B}w_i}$. The minimum is achieved for some connected subgraph with vertex set $B$.
\end{lemma}

Next, as an example of applications of the above bounds, we study the special cases of the complete graph and the ring graph with one stubborn agent. In these cases, Lemma \ref{sinclaire} yields tighter results and it is also easier to use than Lemma \ref{diaconis}.

\begin{example}[Complete graph vs. Ring graph]\label{example}
Assume there is one stubborn agent, node 1, with $K_1>0$. In each case, construct the weighted graph $\hat{\mG}$ and let $\gamma=\{\gamma_i: i \in \mV\}$ be the set of shortest paths from nodes $\mV$ to $u_1$. For the complete graph, the congestion over $(1,u_1)$ is exactly
\ben
 \eta(1,u_1)&=& \frac{1}{K_1}(K_1+(n-1)+2(n-1)^2),
\een
and the congestion over any other edge $(i,1): 1 \leq i \leq n$ is simply
$
 \eta(i,1)= n-1.
$
Hence, for large values of $K_1$, the congestion is dominated by $(1,u_1)$ and for small values of $K_1$, it is dominated by an edge $(i,1)$, for some $1 \leq i \leq n$. More accurately,
$$
T \leq \left\{\begin{array}{ll}
2\frac{K_1+(n-1)+2(n-1)^2}{K_1};& \mbox{ if } K_1 \leq \frac{(n-1)+2(n-1)^2}{n-2}\\
2(n-1);& \mbox{otherwise}
\end{array}\right.
$$
Next, consider a ring graph with odd number of nodes. The congestion over $(1,u_1)$ is
\ben
 \eta(1,u_1)&=& \frac{1}{K_1}(2+K_1+2(2\sum_{i=1}^{(n-1)/2}i)).
\een
Since, for $1 \leq i \leq n$, $\eta(i,1,\gamma)\leq \eta(2,1,\gamma)$, it is enough to find the congestion over the edge $(2,1)$ which is
$
\eta(2,1)= 2\sum_{i=1}^{(n-1)/2}i.
$
This shows that
$$
T \leq \left\{\begin{array}{ll}
\frac{2+K_1+(n^2-1)/2}{K_1};& \mbox{ if } K_1 \leq \frac{8+2(n^2-1)}{n^2-5}\\
\frac{n^2-1}{4};& \mbox{otherwise}
\end{array}\right.
$$
\begin{figure} \label{time}
  \centering
  \includegraphics[width=4 in]{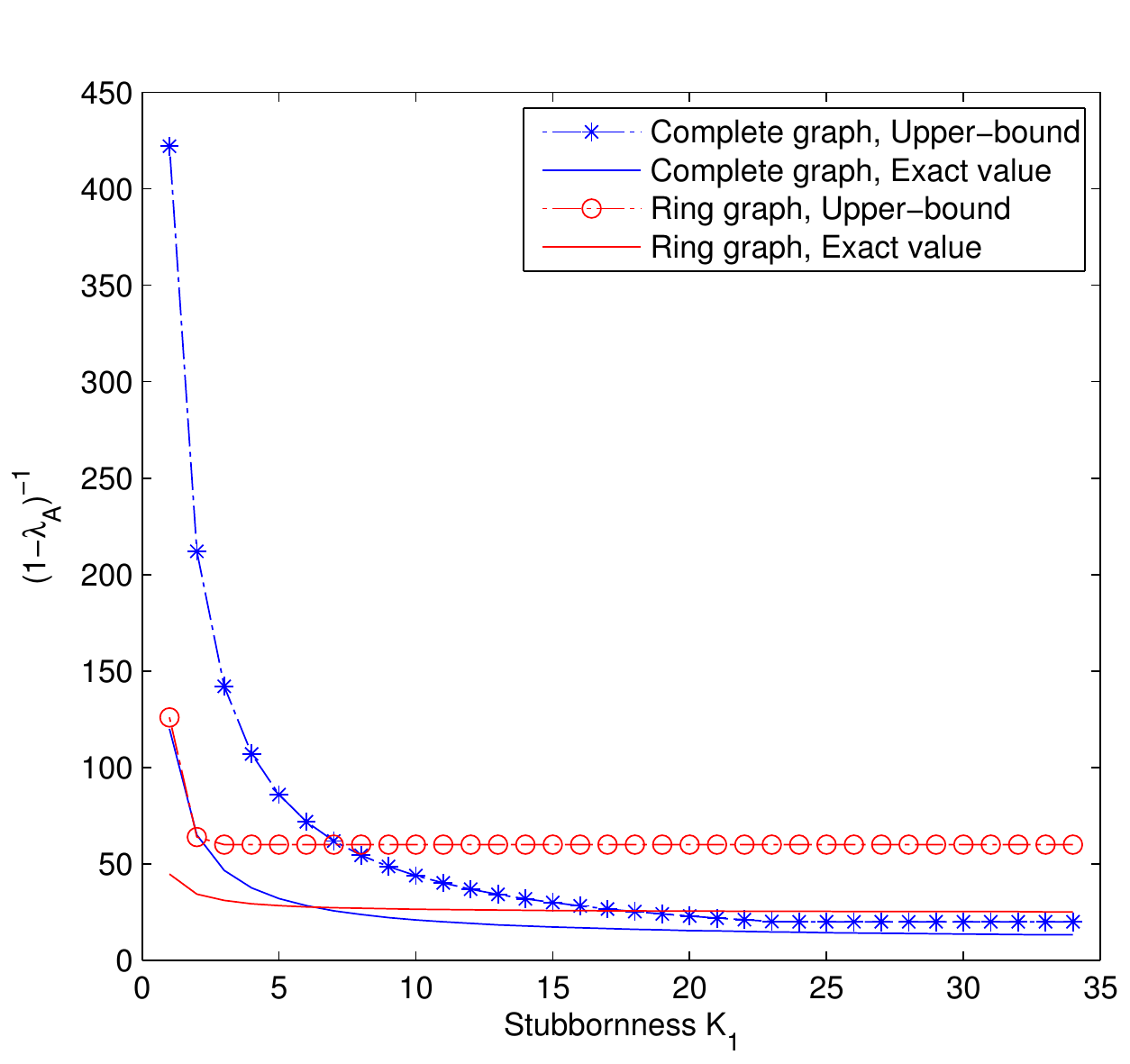}\\
  \caption{The comparison of the convergence time between complete and ring graphs with $n=11$ nodes with one stubborn agent.}\label{complete-vs-ring}
\end{figure}
Figure \ref{time} shows the upper-bound in each graph and compares it with the exact value of $\frac{1}{1-\lambda_A}$ calculated numerically for $n=11$. As $K_1 \to \infty$, the stubborn agent approaches a fully stubborn agent.

We can also compute lower-bounds, based on Lemma \ref{lemma: lower} as follows. The set $B$ that achieves the minimum in \dref{conductance}, either includes node $1$ or not. If $1 \in B$, then it is easy to see that $\min _{B \subseteq \mV: 1 \in B}\psi(B; \hat{\mG})=\psi(\mV; \hat{\mG})=\frac{K_1}{K_1+2 |\mE|}$ in both graphs. If $1 \notin B$, then
$ \min _{B \subseteq \mV: 1 \notin B}\psi(B; \hat{\mG})= \min _{B \subseteq \mV: 1 \notin B}\psi(B; \mG)$.
In the case of the complete graph,
$$\min _{B\subseteq \mV: 1 \notin B}\psi(B; \mG)=\psi(\mV\backslash\{1\}; \mG)=\frac{n-1}{(n-1)^2}=\frac{1}{n-1},$$
and in the case of the ring graph,
$$\min _{B \subseteq \mV: 1 \notin B}\psi(B; \mG)= \psi(\mV\backslash\{1\}; \mG)=\frac{2}{2(n-1)}=\frac{1}{n-1}.$$
Hence, the lower bound is the following. For the complete graph
$$
T \geq \left\{\begin{array}{ll}
\frac{K_1+n(n-1)}{K_1};& \mbox{ if } K_1 \leq \frac{n(n-1)}{(n-2)}\\
n-1;& \mbox{otherwise,}
\end{array}\right.
$$
and for the ring
$$
T \geq \left\{\begin{array}{ll}
\frac{K_1+2n}{K_1};& \mbox{ if } K_1 \leq \frac{2n}{(n-2)}\\
n-1;& \mbox{otherwise.}
\end{array}\right.
$$
\end{example}
\subsection{Canonical bounds via shortest paths}
Let $\gamma=\{\gamma_i: i \in \mV \backslash \mS_F\}$ be the set of \textit{shortest paths} from vertices $\mV \backslash \mS_F $ to a the set $\mS_F \cup u(\mS_P)$, so, in fact, for each $i \in \mV \backslash \mS_F$, $\gamma_i=\gamma_{ij}$ for some $j \in \mS_F \cup u(\mS_P)$. See Figure \ref{short} for an example. Let $\Gamma_j \subseteq \mV \backslash \mS_F$ be the set of nodes that are connected to $j \in \mS_F \cup u(\mS_P)$ via the shortest paths. Also let $|\gamma|:=\max_{i \in \mV \backslash \mS_F } |\gamma_i|$ be the length of maximum of such shortest paths and $|\Gamma|:=\max_{j \in \mS_F \cup u(\mS_P)}|\Gamma_j|$ be the maximum number of nodes connected to any node in $\mS_F \cup u(\mS_P)$. For example, in Figure \ref{short}, $|\Gamma|=4$ and $|\gamma|=3$.
\begin{figure}
  \centering
  \includegraphics[width=4 in]{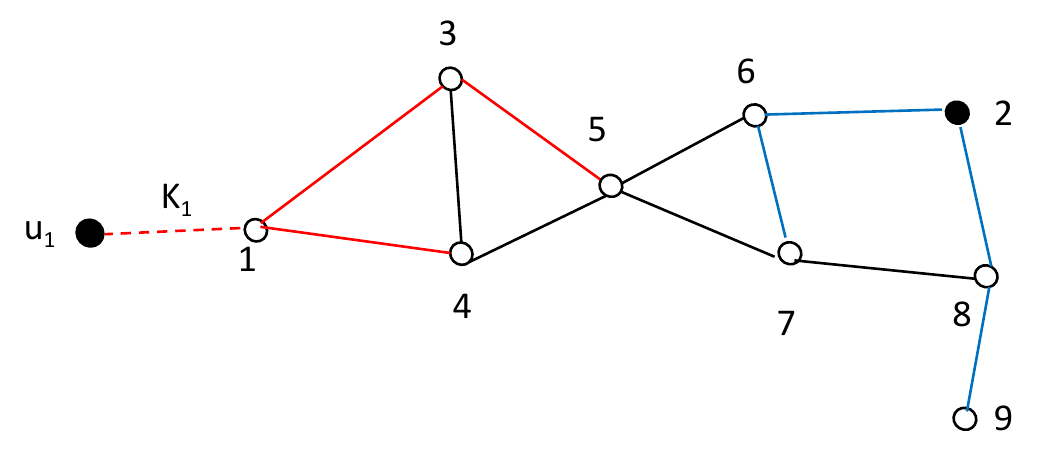}\\
  \caption{A social network consisting of $9$ agents. Vertex $1$ is a partially stubborn agent and vertex $2$ is a fully stubborn agent. The colored edges show the shortest paths from the non-stubborn agents to the set of stubborn agents.}\label{short}
\end{figure}

For each partially stubborn agent $j \in \mS_P$,
\ben
 \eta(j,u_j) & =& \frac{1}{K_j}(K_j+d_j+\sum_{i \in \Gamma_j}d_i |\gamma_i|)\\
 & \leq & 1+\frac{\hat{d}+|\gamma||\Gamma|\tilde{d}}{K_{min}},
\een
where $\tilde{d}:=\max_{i \in \mV \backslash \mS} d_i$ is the maximum degree of non-stubborn agents, $\hat{d}:=\max_{i \in \mS} d_i$ is the maximum degree of stubborn agents, and $K_{min}:=\min_{j \in \mS_P} K_j$ is the minimum stubbornness. Hence, the congestion is dominated by some edge $(j,u_j)$, $j \in \mS_P$, only if the stubbornness $K_j$ is sufficiently small.

It is easy to show that all the paths that pass through an edge $(x,y) \in \mE$ are connected to the same $j \in \mS_F \cup u(\mS_P)$, or equivalently to the same stubborn agent. So for each $(x,y)\in \mE$,
\ben
\eta(x,y)&=& \sum_{i: \gamma_i \ni (x,y)}  d_i |\gamma_i| \leq |\gamma|B \tilde{d},
\een
 where
\be \label{bottleneck}
B:=\max_{(x,y) \in \mE}|\{i: \gamma_i \ni (x,y)\}|,
 \ee
 is the \textit{bottleneck constant}, i.e., the maximum number of shortest paths that pass through any link of the social network. It is clear that $|\Gamma|/\hat{d}\leq B \leq |\Gamma|$ because the maximum bottleneck is at least equal to the bottleneck over an edge directly connected to a stubborn agent. Therefore, for $K_{min} \leq K^*:=\frac{\hat{d}+|\gamma||\Gamma|\tilde{d}}{|\gamma|B\tilde{d}-1}$, $\eta$ is dominated by congestion over some edge $(j,u_j)$, $j \in \mS_P$, and
 \be \label{cong1}
T \leq 2\left(1+\frac{\hat{d}+|\gamma||\Gamma|\tilde{d}}{K_{min}}\right) .
 \ee
For $K_{min}> K^*$, $\eta$ is dominated by an edge of the social network which is the bottleneck, and in this regime
  \be \label{cong2}
T\leq {2|\gamma|B\tilde{d}}.
  \ee
Dependence on $|\gamma|$, in both regimes, intuitively makes sense as it represents the minimum time required to reach any node in the network from stubborn agents.
Hence, the convergence time in general depends on the structure of the social network and the location of the stubborn agents and their levels of stubbornness. There is a dichotomy for high and low levels of stubbornness. For high levels of stubbornness, and in the extreme case of fully stubborn agents, the opinion of the stubborn agent is almost fixed and the convergence time is dominated by the the bottleneck edge and the structure of the social network. For low levels of stubbornness, the transient opinion of stubborn agent may deviate a lot from its equilibrium which could deteriorate the speed of convergence. In fact, for very low levels of stubbornness, this could be the main factor in determining the convergence time. It is worth pointing out that adding more fully stubborn agents, with not necessarily equal initial opinions, or increasing the stubbornness of the agents makes the convergence faster.

\subsection{Scaling Laws}
In this section, we use the canonical bounds to derive scaling laws for the convergence time as the size of the social network $n$ grows.
For any social network, we can consider two cases: (i) There exists no fully stubborn agent, i. e., all the stubborn agents are partially stubborn (ii) At least, one of the agents is fully stubborn.

In both cases, the upper-bound on the convergence time is given by \dref{cong1} and \dref{cong2} depending on the levels of stubbornness of partially stubborn agents. In case (ii), if all the stubborn agents are fully stubborn, then the upper-bound on the convergence time is given by \dref{cong2}.

To get the lower-bounds, we consider the set $B$ in \dref{conductance} to include all the nodes $\mV \backslash \mS_F$. This gives the following lower-bound for the case (i)
\be \label{lowerbound}
T \geq 1+\frac{2|\mE|}{\sum_{j \in \mS_P}K_j},
\ee
and, for the case (ii),
\be \label{lowerbound2}
T \geq \frac{\sum_{j \in \mS_P}K_j+2 |\mE|-\sum_{j \in \mS_F}d_j}{\sum_{j \in \mS_P}K_j+\sum_{j \in \mS_F}d_j}.
\ee

In investigating the scaling laws, the scaling of the number of stubborn agents and their levels of stubbornness with $n$ could play an important role. In the rest of this section, we study scaling laws in graphs with a fixed number of stubborn agents, with fixed levels of stubbornness, as the total number of agent $n$ in the network grows. Then, in any connected graph $\mG$, the smallest possible lower-bound on the convergence time is $T=\Omega(n)$ in the case (i), and $T=\Omega(\frac{|\mE|}{\sum_{j \in \mS_F}d_j})$ in the case (ii) which could be as small as $\Omega(1)$. It is possible to combine the upperbounds \dref{cong1} and \dref{cong2} to obtain a (looser) upper-bound that holds for a fixed number of stubborn agents, with any mixture of partially/fully stubborn agents. Let $d_{max}$ be the maximum degree of the social graph (possibly depending on $n$). The upper-bounds show that $T=O(|\gamma|nd_{max})$ for $K_{min}$ small enough (i.e., smaller than a threshold depending on the structure of the graph) and $T=O(|\gamma|Bd_{max})$ otherwise. Recall that $B$ was the bottleneck constant, and obviously $B < n$, implying that $
T=O(n|\gamma|d_{max}),
$
for a fixed number of stubborn agents consisting of any mixture of partially/fully stubborn agents. Furthermore, it should be clear that $|\gamma|$ is at most equal to the diameter $\delta$ of the graph, hence, as a naive bound,
\be\label{loose cong}
T=O(n\delta d_{max}).
\ee Dependence on the diameter intuitively makes sense as it represents the minimum time required to reach any node in the network from an arbitrary stubborn agent.
\subsection*{Fastest Convergence}
It should be intuitively clear that a graph $\mG$ with a stubborn agent directly connected to $n-1$ non-stubborn agents and with no edges between the non-stubborn agents should have the fastest convergence. In fact, if the stubborn agent is partially stubborn (case (i)), construct $\hat{\mG}$ by adding an extra node $u_1$ and connect it to the stubborn agent $1$ by an edge of weight $K_1$. Then, it is easy to check that $\eta(1,u_1)=1+\frac{3(n-1)}{K_1}$ and $\eta(i,1)=2$, $2 \leq i \leq n$. Hence, $T=O(n)$, and considering the general lower-bound $T=\Omega(n)$, $T=\Theta(n)$ is indeed the sharp order. If the stubborn agent is fully stubborn , then $\eta(i,1)=1$ for all $i$, thus $T \leq 2$ achieving the lower-bound in case (ii).
\subsection*{Complete graph and Ring graph}
For the complete graph, with a fixed set of stubborn agents, $\tilde{d}=\hat{d}=n-1$, $|\Gamma|=O(n)$ and $|\gamma|=2$ and $K^*=O(n)$. If at least one of the agents is partially stubborn, the upper-bound follows from \dref{cong1} which gives $T=O(n^2)$, and, considering the lowerbound \dref{lowerbound}, $T=\Theta(n^2)$ is the right order. If all the stubborn agents are fully stubborn, the upper-bound follows from \dref{cong2} which is $T=O(n)$ because obviously $B=1$ in the complete graph. Using \dref{lowerbound2} gives a lower-bound $\Omega(n)$, so $T=\Theta(n)$ if there is a fixed number of fully stubborn agents.

For the ring network, $\Omega(n)$ is a lower-bound in both cases. To get an upper-bound note that $K^*=O(1)$, however, $B=O(n)$, $\tilde{d}=\hat{d}=2$, $|\Gamma|=O(n)$, and $|\gamma|=O(n)$. Hence, for all fixed levels of stubbornness, and for any fixed number of fully/partially stubborn agents, $T=O(n^2)$.

Figure \ref{time} verifies the results above, as we saw, in the case of one stubborn agent with a fixed $K_1$, and $n$ large enough (larger than a constant depending on the value of $K_1$), the ring network has a faster convergence than the complete graph. For any fixed $n$, and $K_1$ large enough, the complete graph has a faster convergence than the ring.
\subsection*{Expander graphs and Tress}

Expanders are graph sequences such that any graph in the sequence has good expansion property, meaning that there exists $\alpha > 0$ (independent of $n$) such that each subset $S$ of nodes with size $|S| \leq n/2$ has at least $\alpha |S|$ edges to the rest of the network. Expander graphs have found extensive applications in computer science and mathematics (see the survey of \cite{hoory} for a discussion of several applications). An important class of expanders are $d$-regular expanders, where each node has a constant degree $d$. Existence of $d$-regular expanders, for $d >2$, was first established in \cite{pinsker} via a probabilistic argument. There are various explicit constructions of $d$-regular expander graphs, e.g., the Zig Zag construction in \cite{reingold} or the construction in \cite{alon}.

Recall the naive upper-bound \dref{loose cong} when there is a fixed number of (fully/partially) stubborn agents. So, for any bounded degree graph, with maximum degree $d > 2$, and diameter $\delta$, $T=O(n\delta)$. It is easy to see that the diameter of a bounded degree graph, with maximum degree $d$, is at least $\log_{d-1} n$ (Lemma 4.1, \cite{draief}). In fact, for a $d$-regular tree or a $d$-regular expander, $\delta =O(\log n)$ \footnote{To show the latter, consider the lazy random walk over a $d$-regular expander graph, i.e., with transition probability matrix $P=\frac{M}{2d}+\frac{\mI}{2}$ where $M$ is the graph's adjacency matrix. Then, it follows from Cheeger's inequality and the expansion property, that the spectral gap $1-\lambda_2(P) \geq \frac{\alpha^2}{8 d^2}$. Using the relation between the special gap and the diameter $\delta < \frac{\log n}{1-\lambda_2(P)}$ \cite{milman}, we get $\delta \leq \frac{8 d^2}{\alpha^2}\log n$.}. Hence, for these graphs, $T=O(n\log n)$ which is almost as fast as the smallest possible convergence time $\Omega(n)$ when there is at least on partially stubborn agent. When all the stubborn agents are fully stubborn, $T=O(n \log n)$ still holds, by \dref{cong2} because $B=\Theta(n)$ in any bounded degree graph, but, in this case, the convergence is slow compared to the best possible convergence time $\Omega(1)$.
\subsection*{Erdos-Renyi Random Graphs}
Consider an Erod-Renyi random graph with $n$ nodes where each node is connected to any other node with probability $p$, i.e., each edge appears independently with probability $p$. To ensure that the graph is connected, we consider $p=\frac{\lambda \log n}{n}$ for some number $\lambda > 1$. Assume there are a fixed set of stubborn agents with fixed stubbornness parameters. Using the well-known results, the maximin degree of an Erdos-Renyi random graph is $O(\log n)$ with high probability, i.e., with probability approaching to $1$ as $n$ grows~\cite{bollobas}. Also we know that the diameter is $O\left(\frac{\log n}{\log np}\right)=O\left(\frac{\log n}{\log (\lambda\log n) }\right)$ with high probability (in fact, the diameter concentrates only on a few distinct values \cite{chung}). Hence, using the naive upper-bound \dref{loose cong} gives $T=O\left(n\frac{\log^2 n}{\log\log n}\right)$ with high probability. This is very close to the best possible convergence time in case (i) but far from the best possible convergence time in case (ii).
\subsection*{Small world graphs}
The Erdos-Renyi model does not capture many spatial and structural aspects of social networks and, hence, is not a realistic model of social networks \cite{draief}. Motivated by the small world phenomenon observed by Milgram \cite{milgram}, Strogatz-Watts \cite{sw} and Kleinberg \cite{kleinberg} proposed models that illustrate how graphs with spatial structure can have small diameters, thus, providing more realistic models of social networks. We consider a variant of these models, proposed in \cite{draief}, and characterize the convergence time to equilibrium in presence of stubborn agents. We consider two-dimensional graphs for simplicity but results are extendable to the higher dimensional graphs as well.

Start with a social network as a grid $\sqrt{n}\times \sqrt{n}$ of $n$ nodes. Hence, nodes $i$ and $j$ are neighbors if their $l_1$ distance $\|i-j\|=|x_i-x_j|+|y_i-y_j|$ is equal to 1. Following the arguments for the bounded-degree graphs with fixed number of stubborn agents, $T=O(n \delta)$, and in the grid, $\delta=2\sqrt{n}$ obviously, which yields $T=O(n \sqrt{n})$. Note that changing the location of the stubborn agents can change the convergence time only by a constant and does not change the order.

Now assume that each node creates $q$ shortcuts to other nodes in the network. A node $i$ chooses another node $j$ as the destination of the shortcut with probability $\frac{\|i-j\|^{-\alpha}}{\sum_{k \neq i}\|i-k\|^{-\alpha}}$, for some parameter $\alpha >0$. Parameter $\alpha$ determines the distribution of the shortcuts as large values of $\alpha$ produce mostly local shortcuts and small values of $\alpha$ increase the chance of long-range shortcuts. In particular, $q=1$ and $\alpha=0$ recovers the Strogatz-Watts model where the shortcuts are selected uniformly at random. It is shown in \cite{flaxman} that for $\alpha < 2$, the graph is an expander with high probability and hence, using the inequality between the diameter and the spectral gap \cite{milman}, its diameter is of the order of $O(\log n)$ with high probability. We also need to characterize the maximum degree in such graphs. The following lemma is probably known but we were not able to find a reference for it, hence, we have included its proof in Appendix \ref{proof-lemma-degree} for completeness.
\begin{lemma}\label{degree}
Under the small-world network model, $d_{max}=O(\log n)$ with high probability.
\end{lemma}
Hence, putting everything together, using the upper-bound \dref{loose cong}, we get $T=O(n\log^2 n)$. This differs from the smallest possible convergence time in case (i) by a factor of $\log ^2 n$ but far from $\Omega(1)$ in case (ii).
\section{Concluding Remarks}\label{conclusions}
We viewed opinion dynamics as a local interaction game over a social network. When there are no stubborn agents, the best-response dynamics converge to a common opinion in which the impact of the initial opinion of each agent is proportional to its degree. In the presence of stubborn agents, the dynamics converge to an equilibrium in which the opinion of each agent is a convex combination of the initial opinions of the stubborn agents. The coefficients of such convex combination are related to appropriately defined hitting probabilities of the random walk over the social network's graph. An alternative interpretation is based on an electrical network model of the social network where, at equilibrium, the opinion of each agent is simply its voltage in the electrical network.

The bounds on the convergence time in the paper can be interpreted in terms of location and stubbornness levels of stubborn agents, and graph properties such as diameter, degrees, and the so-called bottleneck constant \dref{bottleneck}. The bounds provide relatively tight orders for the convergence time in the case of a fixed number of partially stubborn agents (case (i)) but there is a gap between the lower-bound and the upper-bound when some of the stubborn agents are fully stubborn (case (ii)). Tightening the bounds in case (ii) remains as a future work.

At this point, we discuss the implication of our results in applications where limited advertising budget is to be used to convince a few agents to adopt, for example, a certain opinion about a product/topic. The goal is the optimal selection of such agents to trigger a faster spread of the advertised opinion throughout the social network. This, in turn, implies that, over a finite time, more agents will be biased towards the advertised opinion. Similar \textit{leader selection} problems have been discussed in~\cite{borkar, clark, yildiz} for different applications where the goal is to select a set of $M$ agents with fixed states to optimize the system performance such as minimizing the convergence time to consensus~\cite{borkar}, minimizing the error when the observations are noisy~\cite{clark}, or maximizing the impact of stubborn agents on the long-run expected opinions of the agents \cite{yildiz}. The common methodology is to show that the objective is a sub-modular set function and use the sub-modular optimization framework in, e.g., \cite{modular}, to produce a greedy procedure where agents are added according to a greedy sequence. Although the greedy algorithm is useful, it does not answer the question in its simplest form $M=1$, and may involve the inversion of typically large matrices \cite{borkar} in social networks with very large number of users.

Using the simple bound \dref{cong2}, the question is reduced to where to place a fixed number of fully-stubborn agents in order to minimize $|\gamma|B\tilde{d}$. Recall that $\gamma$ is the maximum length of shortest paths from non-stubborn agents to stubborn agents, $B$ is the bottleneck constant defined in \dref{bottleneck}, and $\tilde{d}$ is the maximum degree of non-stubborn agents. Since empirical graphs of social networks exhibit small-world network characteristics, $|\gamma|$ is already very small (it is less than the diameter of the graph which is already a $\log n$ quantity), and hence, the product $B \cdot \tilde{d}$ is the dominating factor. It is believed that degree distribution in many networks, such as social networks, Internet topology, WWW induced graph, Hollywood graph, etc, follows a power-law distribution (see~\cite{newman} for a survey with more examples). Heuristically, when there are a few very high degree nodes and most of the nodes are of low degrees, selecting the high degree nodes, as stubborn agents, reduces $\tilde{d}$ dramatically and also reduces $B$ because many agents (the neighbors of the stubborn agents) are now directly connected to the stubborn agents. On the other hand, selecting the possibly low degree nodes of bottleneck edges as stubborn agents reduces $B$ but this reduction is at most by a factor equal to low degrees of such nodes. Therefore, in general, the high degree nodes seem to be good candidates for placement of stubborn agents. It will be certainly interesting to establish the validity of such a heuristic more rigorously.
\appendices\label{proofs}
\section{Proof of Lemma \ref{conv1}}\label{proof-lemma-conv1}
A similar result as Lemma \ref{conv1} is standard when analyzing the convergence of the probability distribution of a Markov chain. Opinion dynamics are different than probability evolution in Markov chains, but the same ideas as in the case of probability distributions also work when analyzing the convergence of opinion dynamics. We present the proof of Lemma \ref{conv1} here for completeness.

We can compute the error recursively as follows.
\be
e(t+1)&=&Ax(t)-\mathds{1}_n\pi^Tx(0)-\mathds{1}_n\pi^Tx(t)+\mathds{1}_n\pi^Tx(t)\\
& =& Ax(t)-A\mathds{1}_n\pi^Tx(0)-\mathds{1}_n\pi^Tx(t)+\mathds{1}_n\pi^TA^tx(0)\\
& =& Ax(t)-A\mathds{1}_n\pi^Tx(0)-\mathds{1}_n\pi^Tx(t)+\mathds{1}_n\pi^Tx(0)\\
& =& Ax(t)-A\mathds{1}_n\pi^Tx(0)-\mathds{1}_n\pi^Tx(t)+\mathds{1}_n\pi^T\mathds{1}_n\pi^T x(0)\\
& =& (A-\mathds{1}_n\pi^T)(x(t)-\mathds{1}_n\pi^Tx(0))\\
& =& (A-\mathds{1}_n\pi^T) e(t).
\ee
$(A, \pi)$ is reversible, thus $D^{1/2}AD^{-1/2}$ is a symmetric matrix where $D=\mathrm{diag}(\pi)$. Then, it is well-known, e.g., see \cite{pier}, that $A$ has real eigenvalues and $n$ distinct right eigenvectors $v_1,v_2,\cdots,v_n$ and $n$ distinct left eigenvectors $u_1,u_2,\cdots u_n$ such that $u_i=D v_i$. So, it follows from orthogonality of left and right eigenvectors that $\langle v_i, v_j \rangle_{\pi}=\delta_{ij}$, where $\delta_{ij}=1$ if $i=j$ and is zero otherwise. Therefore, we have
\ben
(A-\mathds{1}_n\pi^T)\mathds{1}_n=\mathds{1}_n-\mathds{1}_n=0
\een
and, for $i\geq 2$,
\ben
(A-\mathds{1}_n\pi^T)v_i=\lambda_iv_i-\mathds{1}_n\pi^Tv_i=\lambda_iv_i,
\een
Following standard line of arguments as in \cite{pier}, $e(t)=\sum_{i=1}^n \langle e(t),v_i \rangle_{\pi}v_i$.
therefore,
\ben
(A-\mathds{1}_n\pi^T) e(t)=\sum_{i=2}^n\lambda_i \langle e(t),v_i \rangle_{\pi}v_i.
\een
and
\ben
\|e(t+1)\|^2_{\pi}&=&\sum_{i=2}^n\lambda_i^2 \langle e(t),v_i \rangle_{\pi}^2\|v_i\|^2_{\pi}\\
& = &\sum_{i=2}^n\lambda_i^2 \langle e(t),v_i \rangle_{\pi}^2\\
& \leq & \rho_2^2 \sum_{i=2}^n\langle e(t),v_i \rangle_{\pi}^2\\
& =& \rho_2^2 \|e(t)\|^2_{\pi},
\een
where $\rho_2:=\max_{i\neq 1} {|\lambda_i|}$ is the SLEM of $A$. So
$
\|e(t+1)\|_{\pi} \leq \rho_2 \|e(t)\|_{\pi},
$
and, accordingly,
$
\|e(t)\|_{\pi}\leq \rho_2^t \|e(0)\|_{\pi}.
$
\section{Proof of Lemma \ref{lemma: convex}}\label{proof-lemma-convex}
The transition probability matrix of the random walk over $\hat{\mG}$ is given by
\be \label{P matrix}
P=\left[\begin{array}{ll}
\hat{A}_{n \times n} & \hat{B}_{n \times |\mS_P|} \\
\mI_{|\mS_P|} & 0
\end{array}\right].
\ee
$\mI_{|\mS_P|}$ is the identity matrix of size $|\mS_P|$, i.e., when the walk reaches $u_i$, it returns to its corresponding stubborn agent $i$ with probability $1$. Nonzero elements of $\hat{A}$ correspond to transitions between vertices of $\mV$. Nonzero elements of $\hat{B}$ correspond to transitions from a partially stubborn agent $i \in \mS_P$ to $u_i$. The matrices $\hat{A}$ and $A$ only differ in the rows corresponding to agents $\mS_F$ which are all-zero rows in $A$. Notice that $x_i(t)=x_i(0)$ for all $i \in \mS_F$ and $t \geq 0$. Hence, we can focus on the dynamics of $\tilde{x}(t)=[x_i(t): i \in \mV \backslash \mS_F]^T$.

Let $\tilde{A}$ be the matrix obtained from $\hat{A}$ (or $A$) by removing rows and columns corresponding to fully stubborn agents $\mS_F$. Let $\hat{A}_{\mS_F}$ denote the columns of $\hat{A}$ corresponding to $\mS_F$. Let $\tilde{B}$ be the matrix obtained from $B$ by (i) replacing the columns corresponding to fully stubborn agents $\mS_F$ with $\hat{A}_{\mS_F}$ (or ${A}_{\mS_F}$), (ii) removing rows corresponding to $\mS_F$, (iii) removing the columns corresponding to non-stubborn agents (which are all zero columns). Then, we have
$$
\tilde{x}(t+1)=\tilde{A}\tilde{x}(t)+\tilde{B}x_{\mS}(0).
$$
where $x_{\mS}(0)=[x_i(0): i \in S]^T$. Note that both $A$ and $\tilde{A}$ have the same largest eigenvalue, i.e., $\lambda_A=\lambda_{\tilde{A}}$. The dynamics converge to the equilibrium $\tilde{x}(\infty)=(\mI-\tilde{A})^{-1}\tilde{B}x_{\mS}(0)$.

For each vertex $i \in \mV$, and $j \in \mS_F$, let $F_{ij}:=\mathbb{P}_i(\tau=\tau_{j})$ be the probability that random walk hits $j$ first, among vertices in $\mS_F \cup u(\mS_P)$, given the random walk starts from vertex $i$. Also, for each vertex $i \in \mV$, and $u_j \in u(\mS_P)$, let $F_{ij}:=\mathbb{P}_i(\tau=\tau_{u_j})$ be the probability that random walk hits $u_j$ first, among vertices in $\mS_F \cup u(\mS_P)$, given the random walk starts from vertex $i$. Then, we have the following recursive formulas for the $F_{ij}$ probabilities. For every $i \in \mV \backslash \mS_F$ and every $j \in S_F$,
\be
F_{ij}=\hat{A}_{ij}+\sum_{k \in \mV \backslash \mS_F} \hat{A}_{ik}F_{kj},
\ee
and for every $i \in \mV \backslash \mS_F$ and every $j \in \mS_P$,
\be
F_{ij}=\hat{B}_{ij}+\sum_{k \in \mV \backslash \mS_F} \hat{A}_{ik}F_{kj}.
\ee
Note that $\tilde{B}$ is $[\hat{B} \hat{A}_{\mS_F}]$ without the rows corresponding to $\mS_F$. Hence, putting the two equations together in the matrix form, $F=\tilde{B}+\tilde{A}F$ or $F=(\mI-\tilde{A})^{-1}\tilde{B}$.

Note that for any $i \in \mS_F$, $F_{ii}=1$ and $x_i(t)=x_i(0)$ at all times $t \geq 0$. Hence, the equilibrium at each node $i\in \mV$, is a \textit{convex combination of initial opinions of stubborn agents}, where
\be
x_i(\infty)=\sum_{j \in \mS} F_{ij}x_j(0).
\ee
\section{Proof of Lemma \ref{elec}}\label{proof-lemma-elec}
Recall graph $\hat{\mG}$ with edge weights $\{w_{ij}: (i,j) \in \hat{\mE}\}$. By \dref{individual dynamics}, and taking the limit as $t \to \infty$, the equilibrium is the solution to the following set of linear equations
\be \label{equi3}
x_i(\infty)=\frac{1}{w_i}\sum_{j \in \partial_i}w_{ij}x_j(\infty),
\ee
for each node $i \in \hat{\mV}$, with boundary conditions $x_{u_i}(\infty)=x_i(0)$, $i \in \mS_P$, and $x_{i}(\infty)=x_i(0)$ for $i \in \mS_F$. Now assume each edge $(i,j)\in \hat{\mE}$ has a conductance $w_{ij}$ and vertices $\mS_F \cup u(\mS_P)$ are voltage sources where the voltage of each source $i \in \mS_F$ is $x_i(0)$ volts and the voltage of each source $u_j \in  u(\mS_P)$, $j \in \mS_P$, is $x_j(0)$ volts. Let $v_i$ be the voltage of node $i$. Kirchhoff's current law states that the total current entering each node must be zero, i.e., for each node $i \in \mV \backslash \mS_F$,
$
\sum_{j \in \partial_i}w_{ij}(v_i-v_j)=0
$
or equivalently,
\be
w_i v_i=\sum_{j \in \partial_i}w_{ij}v_j
\ee
which, comparing to \dref{equi3}, shows that $x_i(\infty)=v_i$. Note that having a fully stubborn agent $i$, with $K_i=\infty$, corresponds with connecting $i$ to a fixed voltage of $x_i(0)$ volts with an edge of infinite conductance (short circuit). Hence, $K_i$'s can be interpreted as the internal conductance of the voltage sources. A fully stubborn agent $i$ with $K_i =\infty$ corresponds to an ideal voltage source with zero internal resistance.
\section{Proof of Lemma \ref{conv2}}\label{proof-lemma-conv2}
From the definition of $e(t)$,
\ben
e(t)&=&A^tx(0)+\sum_{s=0}^{t-1}A^{s}Bx(0)-\sum_{s=0}^{\infty}A^{s}Bx(0)\\
&=& A^tx(0)-\sum_{s=t}^{\infty}A^{s}Bx(0)\\
& =& A^t \left(x(0)-\sum_{s=0}^{\infty}A^{s}Bx(0)\right)
\een
Hence $e(t+1)=Ae(t)$. Let $\lambda_A$ denote the largest eigenvalue of the irreducible sub-stochastic matrix $A$.
Trivially $e_i(t)=0$ for all fully stubborn agents $i \in \mathcal{S}_F$. Let $\tilde{e}(t):=(e_i(t): i \in \mV \backslash \mS_F)^T$ denote the vector of errors without the fully stubborn agents. Then $\tilde{e}(t)=\tilde{A}\tilde{e}(t-1)$ holds, where $\tilde{A}$ is the matrix obtained from $A$ by removing rows and columns corresponding to agents $\mathcal{S}_F$. Note that $\tilde{A}$ and $A$ have the same largest eigenvalue, i.e., $\lambda_A=\lambda_{\tilde{A}}$.

Consider the Markov chain defined by $P$ in (\ref{P matrix}). It is easy to check that $P$ is reversible with respect to a distribution $\pi=(\pi_i=\frac{w_i}{Z}: i \in \hat{\mV})^T$ where $w_i$ is the weighted degree of vertex $i$, given by \dref{weight}, and $Z=2 (|\mE|+ \sum_{i \in \mS_P}K_i)$ is the normalizing constant\footnote{By definition of reversibility, $\pi_iP_{ij}=\pi_j P_{ji}$ for all $i,j \in \hat{\mV}$}.
Note that $\pi_i\tilde{A}_{ij}=\pi_j\tilde{A}_{ji}$ holds for all $i,j \in \mV \backslash \mS_F$. By minor abuse of terminology, we would also call $\tilde{A}$ reversible with respect to the distribution $\tilde{\pi}=\left(\pi_i/\pi(\tilde{A}): i \in \mV \backslash \mS_F\right)^T$, where $\pi(\tilde{A})$ is the normalization constant. Let $\tilde{D} =\mathrm{diag}(\tilde{\pi})$. Then, using the same trick as in the characterization of eigenvalues of a reversible stochastic matrix, $A^*=\tilde{D}^{1/2}\tilde{A}\tilde{D}^{-1/2}$ is symmetric and has the same (real) eigenvalues as $\tilde{A}$. Moreover $A^*$ is diagonalizable with a set of equal right and left eigenvectors $\theta_1, \cdots, \theta_{n-|\mS_F|}$. Correspondingly, if $u_1, \cdots, u_{n-|\mS_F|}$ denote the left eigenvectors of $\tilde{A}$ and $v_1, \cdots, v_{n-|\mS_F|}$ denote its right eigenvectors, it should hold that $u_i=\tilde{D}v_i$. Also from the orthogonality of $\theta_i's$, we have $\langle u_i, u_j \rangle _{1/\tilde{\pi}}=\delta_{ij}$ and $\langle v_i, v_j \rangle _{\tilde{\pi}}=\delta_{ij}$. Using $\{v_1, \cdots, v_{n-|\mS_F|}\}$ as a base for $ \mathds{R}^{n-|\mS_F|}$, $\tilde{e}(t)$ can be expressed as
\ben
\tilde{e}(t)=\sum_{i=1}^{{n-|\mS_F|}}\langle \tilde{e}(t),v_i \rangle_{\tilde{\pi}}v_i,
\een
so
$$
\tilde{A}\tilde{e}(t)=\sum_{i=1}^n \lambda_i \langle \tilde{e}(t),v_i \rangle_{\tilde{\pi}}v_i.
$$
Therefore,
\ben
\|\tilde{e}(t+1)\|^2_{\tilde{\pi}}&=&\sum_{i=1}^n\lambda_i^2 \langle \tilde{e}(t),v_i \rangle_{\tilde{\pi}}^2\|v_i\|^2_{\tilde{\pi}}\\
& = &\sum_{i=1}^n\lambda_i^2 \langle \tilde{e}(t),v_i \rangle_{\tilde{\pi}}^2\\
& \leq & \lambda_A^2 \sum_{i=1}^n\langle \tilde{e}(t),v_i \rangle_{\tilde{\pi}}^2\\
& =& \lambda_A^2 \|\tilde{e}(t)\|^2_{\tilde{\pi}},
\een
So $
\|\tilde{e}(t+1)\|_{\tilde{\pi}} \leq \lambda_A \|\tilde{e}(t)\|_{\tilde{\pi}}.
$
Accordingly, $
\|\tilde{e}(t)\|_{\tilde{\pi}}\leq \lambda_A^t \|\tilde{e}(0)\|_{\tilde{\pi}}.
$
\section{Proofs of Lemmas \ref{diaconis}, \ref{sinclaire}, and \ref{lemma: lower}}\label{proof-lemmas}
The three Lemmas are based on the extremal characterization of the eigenvalues. First, we present an extremal characterization for the largest eigenvalue of a sub-stochastic (and reversible) matrix. Then, we state the proofs of individual lemmas.

$A$ and $\tilde{A}$ have the same largest eigenvalue (recall that $\tilde{A}$ is obtained from $A$ by removing rows and columns corresponding to fully stubborn agents $\mS_F$, as in Appendix \label{proof-lemma-convex}.). Consider the Markov chain defined by $P$ in (\ref{P matrix}). $P$ is reversible with respect to $\pi=(\pi_i=\frac{w_i}{Z}: i \in \hat{\mV})^T$ where $w_i$ is the weighted degree of vertex $i$, given by \dref{weight}, and and $Z=2 (|\mE|+ \sum_{i \in \mS_P}K_i)$ is the normalizing constant. Note that $\pi_i\tilde{A}_{ij}=\pi_j\tilde{A}_{ji}$ holds for all $i,j \in \mV \backslash \mS_F$. By minor abuse of terminology, we would also call $\tilde{A}$ reversible with respect to the distribution $\tilde{\pi}=\left(\pi_i/\pi(\tilde{A}): i \in \mV \backslash \mS_F\right)^T$, where $\pi(\tilde{A})$ is the normalization constant. Thus, it follows from extremal characterization of eigenvalues \cite{horn, pier} that
$$
1-\lambda_A=\inf_{f \neq 0} \frac{\langle(\mI-\tilde{A})f,f\rangle_{\tilde{\pi}}}{\langle f,f\rangle_{\tilde{\pi}}}
$$
where the infimum is over all functions $f:\mV \backslash \mS_F \to \mathds{R}$. The above characterization can also be written as
$$
1-\lambda_A=\inf_{g \neq 0} \frac{\langle(\mI-\hat{A})g,g\rangle_{{\pi}}}{\langle f,f\rangle_{{\pi}}}
$$
where now the infimum is over all functions $g: \mV \to \mathds{R}$ with $g(\mS_F)=0$.
%
Recall the random walk over the weighted graph $\hat{\mG}(\hat{\mV}, \hat{\mE})$ with transition probability matrix $P$ \dref{P matrix}. Equivalently, we can also write the characterization as
$$
1-\lambda_A=\inf_{\phi \neq 0} \frac{\langle(\mI-{P})\phi,\phi\rangle_{\pi}}{\langle \phi,\phi\rangle_{\pi}}.
$$
where now the infimum is over functions $\phi:\hat{\mV} \to \mathds{R}$, such that $\phi\left(\mS_F \cup u(\mS_P)\right)=0$  \footnote{It is worth pointing out that function $f$, accordingly $g$ or $\phi$, that achieves the infimum is the right eigenvector corresponding to $\lambda_A$ and from Peron-Ferobenius theorem, it must be nonnegative.}.
Then, $\langle(\mI-{P})\phi,\phi\rangle_{\pi}=\mE(\phi,\phi)$ where $\mE(\phi,\phi)$ is the Dirichlet form
$$
\mE(\phi,\phi)=\frac{1}{2}\sum_{i,j \in \hat{\mV}}\pi_i {P}_{ij}(\phi(i)-\phi(j))^2
$$
which, in terms of the edge weights of $\hat{\mG}$, is equal to
$$
\mE(\phi,\phi)=\frac{1}{2w}\sum_{i,j \in \hat{\mV}}w_{ij}(\phi(i)-\phi(j))^2.
$$
where $w:=\sum_{i \in \hat{\mV}}w_i$. Similarly,
$$
\langle \phi,\phi\rangle_{\pi}=\frac{1}{w}\sum_{i \in \mV \backslash \mS_F}w_i\phi^2(i).
$$
Define a path from vertex $i$ to vertex $j$, as a collection of \textit{oriented} edges that connect $i$ to $j$. For any vertex $i \in \mV \backslash \mS_F$, consider a path $\gamma_{i}$ from $i$ to the set $\mS_F\cup u(\mS_P)$ that does not intersect itself, i.e., $\gamma_i=\{(i,i_1), (i_1,i_2), \cdots, (i_m,j)\}$ for some $j \in \mS_F\cup u(\mS_P)$. Then, we can write $\phi(i)=\sum_{(x,y)\in \gamma_i} (\phi(x)-\phi(y))$.
\begin{proof}[Proof of Lemma \ref{diaconis}]
The result follows from the extremal characterization of $1-\lambda_A$. Note that
\be
\langle \phi,\phi\rangle_{\pi} & = & \frac{1}{w}\sum_{i \in \mV \backslash \mS_F}w_i\left(\sum_{(x,y)\in \gamma_i} (\phi(x)-\phi(y))\right)^2 \nonumber \\
& =& \frac{1}{w}\sum_{i \in \mV \backslash \mS_F}w_i\left(\sum_{(x,y)\in \gamma_i} \frac{1}{\sqrt{w_{xy}}}\sqrt{w_{xy}}(\phi(x)-\phi(y))\right)^2 \nonumber  \\
& \leq & \frac{1}{w}\sum_{i \in \mV \backslash \mS_F}w_i\left(\sum_{(x,y)\in \gamma_i} \frac{1}{w_{xy}}\right)\left(\sum_{(x,y)\in \gamma_i} w_{xy}(\phi(x)-\phi(y))^2\right) \label{cauchy1}
\ee
\be
& =& \frac{1}{w}\sum_{i \in \mV \backslash \mS_F}w_i|\gamma_i|_w\left(\sum_{(x,y)\in \gamma_i} w_{xy}(\phi(x)-\phi(y))^2\right) \nonumber  \\
&=& \frac{1}{w}\sum_{x,y \in \hat{\mV}} w_{xy}(\phi(x)-\phi(y))^2\left(\sum_{i: \gamma_i \ni (x,y) } w_i|\gamma_i|_w\right) \nonumber  \\
& \leq &2 \mE(\phi,\phi) \xi
\ee
This concludes the proof. Inequality \dref{cauchy1} is an application of Cauchy-Schwarz inequality.
\end{proof}
\begin{proof}[Proof of Lemma \ref{sinclaire}]
The proof is again based on the extremal characterization. Note that
\be
\langle \phi,\phi\rangle_{\pi} & = & \frac{1}{w}\sum_{i \in \mV \backslash \mS_F}w_i\left(\sum_{(x,y)\in \gamma_i} (\phi(x)-\phi(y))\right)^2 \nonumber \\
& \leq & \frac{1}{w}\sum_{i \in \mV \backslash \mS_F}w_i|\gamma_i|\sum_{(x,y)\in \gamma_i} (\phi(x)-\phi(y))^2 \label{cauchy2}\\
& =& \frac{1}{w}\sum_{x,y \in \hat{\mV}}(\phi(x)-\phi(y))^2\sum_{i:\gamma_i \ni(x,y)} w_i|\gamma_i| \nonumber\\
& =& \frac{1}{w}\sum_{x,y \in \hat{\mV}}w_{xy}(\phi(x)-\phi(y))^2\frac{1}{w_{xy}}\sum_{i:\gamma_i \ni(x,y)} w_i|\gamma_i| \nonumber \\
& \leq & 2 \mE(\phi,\phi) \eta.
\ee
which concludes the proof. Inequality \dref{cauchy2} follows from Cauchy-Schwarz inequality.
\end{proof}
\begin{proof}[Proof of Lemma \ref{lemma: lower}]
To find an upper bound on $1-\lambda_A$, consider indicator functions of the form $\mathds{1}_B(i)$, $B \subseteq \mV\backslash \mS_F$, in the extremal characterization of eigenvalues. Then, we have
 \ben
 1-\lambda_A &\leq & \frac{\mE(\mathds{1}_B,\mathds{1}_B)}{\langle \mathds{1}_B, \mathds{1}_B \rangle_\pi}\\
 & =&  \frac{\sum_{i \in B, j \notin B}w_{ij}}{\sum_{i \in B}w_i} =:\psi(B; \hat{\mG})
 \een
 And accordingly,
 $
 1-\lambda_A \leq  \min_{B \subseteq \mV\backslash \mS_F}\psi(B; \hat{\mG}).
 $
 It is easy to see that the minimizing $B$ is the vertex set of a connected subgraph of $\mG \backslash \mS_F $ in the above minimization.
\end{proof}
\section{Proof of Lemma \ref{degree}}\label{proof-lemma-degree}
First, we show that the mean degree is less than a constant for all $n$ for all $\alpha \geq 0$. 
\ben
\sum_{k \neq i}\|i-k\|^{-\alpha} & \geq & \sum_{l=1}^{\sqrt{n}/2} (l)l^{-\alpha} \geq  \int_1^{\sqrt{n}/2} x^{1-\alpha}dx =\frac{1}{2 -\alpha}((\sqrt{n}/2)^{2-\alpha}-1)\\
& \geq & \frac{(\sqrt{n}/2)^{2-\alpha}}{2 -\alpha},\ \mbox{ for }\alpha \neq 2,
\een
and $\log(\sqrt{n}/2)$ for $\alpha=2$.
For any node $u$, $d_u \leq 4 + q + \sum_{i=1, i \neq u}^n \sum_{s=1}^q X_{is}$, where $X_{is}$ is a Bernoulli random variable indicating if the s-th shortcut, $1 \leq s \leq q$, from $i$ is connected to $u$ or not. Hence $d_u$ is the summation of $(n-1)q$ independent random variables which are not necessarily identically distributed except those random variables that correspond to shortcuts from the same node. Let $\bar{d_u}:=\expect{d_u}$. Hence
\ben
\bar{d_u} &\leq& 4 + q +q \sum_{l=1}^{2\sqrt{n}} (4l)l^{-\alpha}\frac{2 -\alpha}{(\sqrt{n}/2)^{2-\alpha}} \\
& \leq &4 + q +q \frac{2 -\alpha}{(\sqrt{n}/2)^{2-\alpha}} \Big(1+\int_1^{2\sqrt{n}+1} x^{1-\alpha}dx\Big)\\
& = & 4 + q +q \frac{2 -\alpha}{(\sqrt{n}/2)^{2-\alpha}}\Big(1+\frac{1}{2-\alpha}(2\sqrt{n}+1)^{2-\alpha}-1)\Big)\\
& \leq & 4 + q +2q  6^{2-\alpha} \leq 77 q
\een
and for $\alpha=2$,
\ben
\bar{d_u}& \leq &4 + q +q \frac{1}{\log (\sqrt{n}/2)} \Big(1+\int_1^{2\sqrt{n}+1} x^{-1}dx\Big)\\
& = & 4 + q +q \frac{1}{\log (\sqrt{n}/2)}\Big(1+\log(2\sqrt{n}+1)\Big)\\
& \leq & 4 + q +4q \leq 9q
\een
Next, we use the following version of the checkoff bound for summation of independent, but not identically distributed, Bernoulli random variables (Lemma 2.4 of \cite{draief}),
\ben
\prob{d_u-\bar{d_u}\geq \epsilon \bar{d_u}} \leq \exp(-h(\epsilon)\bar{d}_u)
\een
where $h(x):=(1+x)\log (1+x)-x$. Let $\epsilon = \frac{\log n}{\bar{d}_u}$, then $h(\epsilon) \geq \frac{\log n}{\bar{d}_u} (\log \frac{\log n}{\bar{d}_u}-1)$, thus,
\ben
\prob{d_u-\bar{d_u}\geq \epsilon \bar{d_u}} \leq n^{1-\log \frac{\log n}{\bar{d}_u}}\leq n^{1-\log \frac{\log n}{77q}}=n^{1+\log (77q)-\log \log n}.
\een
By the union bound
\ben
\prob{\exists u: d_u-\bar{d_u}\geq \epsilon \bar{d_u}} \leq n^{2+\log (77q)-\log \log n}.
\een
so, as $n \to \infty$, $d_{max} =O(\log n)$ with high probability.

\end{document}